\documentclass[10pt, conference, letterpaper]{ IEEEtran}
\IEEEoverridecommandlockouts
\usepackage{tikz}
\usetikzlibrary{arrows.meta, positioning, shapes.multipart, fit, calc}

\usepackage{cite}

\usepackage{amsmath,amssymb,amsfonts}
\usepackage{url}
\usepackage{multirow}
\usepackage{cleveref}
\crefname{figure}{Fig.}{Fig.} 

\crefname{table}{Table}{Tables} 
\crefname{equation}{Equation}{Equations} 
\crefname{section}{section}{Section} 
\usepackage{array}
\usepackage{tabularx}
\usepackage{algorithm}
\usepackage{algpseudocode}
\usepackage{amsmath,amsthm}
\usepackage{subcaption}
\newtheorem{theorem}{Theorem}
\usepackage{graphicx}
\usepackage{textcomp}
\usepackage{xcolor}
\usepackage{booktabs}
\usepackage{siunitx}

\def\BibTeX{{\rm B\kern-.05em{\sc i\kern-.025em b}\kern-.08em
    T\kern-.1667em\lower.7ex\hbox{E}\kern-.125emX}}
\begin{document}

\title{ATRO: A Fast Algorithm for Topology Engineering of Reconfigurable Datacenter Networks
}
\author{
    \IEEEauthorblockN{
        Yingming Mao\IEEEauthorrefmark{1}\IEEEauthorrefmark{2}, 
        Qiaozhu Zhai\IEEEauthorrefmark{1}, 
        Ximeng Liu\IEEEauthorrefmark{3}\IEEEauthorrefmark{4}, 
        Xinchi Han\IEEEauthorrefmark{3}, 
        Fanfan Li\IEEEauthorrefmark{1}, \\
        Shizhen Zhao\IEEEauthorrefmark{3}, 
        Yuzhou Zhou\IEEEauthorrefmark{1}, 
        Zhen Yao\IEEEauthorrefmark{5}, 
        and Xia Zhu\IEEEauthorrefmark{5}
    }
    \IEEEauthorblockA{\IEEEauthorrefmark{1}Xi'an Jiaotong University \quad \IEEEauthorrefmark{2}Shanghai Innovation Institute\\\IEEEauthorrefmark{3}Shanghai Jiao Tong University \quad  \IEEEauthorrefmark{4}Zhongguancun Academy
    \IEEEauthorrefmark{5}Huawei}
}

\maketitle
\insert\footins{\noindent\footnotesize 
Yingming Mao and Qiaozhu Zhai contributed equally to this work. \\
Corresponding author: Qiaozhu Zhai (qzzhai@sei.xjtu.edu.cn). \\
This work was supported by the Fundamental Research Funds for the Central Universities (xzy022025039) and by the Fundamental Research Funds for the Central Universities (xtr052025030).}
\begin{abstract}
Reconfigurable data center networks (DCNs) enhance traditional architectures with optical circuit switches (OCSs), enabling dynamic reconfiguration of inter-PoD links, i.e., the logical topology. Optimizing this topology is crucial for adapting to traffic dynamics but is challenging due to its combinatorial nature. The complexity increases further when demands can be distributed across multiple paths, requiring joint optimization of topology and routing. We propose \textbf{Alternating Topology and Routing Optimization (ATRO)}, a unified framework that supports both \emph{one-hop topology optimization} (where traffic is routed via direct paths) and \emph{multi-hop joint optimization} (where routing is also optimized). Although these settings differ in constraints, both are combinatorially hard and challenge solver-based methods. ATRO addresses both cases efficiently: in the one-hop case, it guarantees the global optimum via an accelerated binary search; in the multi-hop case, it alternates between topology and routing updates, with routing steps optionally accelerated by existing traffic engineering (TE) methods. ATRO supports warm-starting and improves solution quality monotonically across iterations. ATRO remains competitive even when paired with solver-free TE methods, forming a fully solver-free optimization pipeline that still outperforms prior approaches in runtime and maximum link utilization across diverse workloads.

\end{abstract}


\section{Introduction}

The explosive growth of social networks and the rapid evolution of large language models (LLMs)—which demand massive GPU clusters—are placing unprecedented pressure on data center infrastructures~\cite{qianAlibabaHPNData2024}. In response, operators are scaling up their networks and constructing increasingly larger data centers~\cite{yuOrcaDistributedServing2022}. To improve scalability and efficiency, many operators are turning to optical circuit switches (OCSs), which enable reconfigurable and dynamic logical topologies~\cite{farringtonHeliosHybridElectrical2010}.

Reconfigurable DCNs adapt logical topologies to traffic patterns via two paradigms: \emph{multi-hop} and \emph{one-hop}. Multi-hop designs, e.g., Jupiter Evolving~\cite{poutievskiJupiterEvolvingTransforming2022}, rely on dynamic routing over relatively static structures. Conversely, one-hop designs~\cite{melletteRotorNetScalableLowcomplexity2017,ballaniSirius2020} reconfigure OCSs to directly connect Points of Delivery (PoDs), often bypassing routing decisions to minimize latency~\cite{farringtonHeliosHybridElectrical2010}. Consequently, \textit{topology optimization (TO)} dominates one-hop settings, while \textit{joint topology and routing optimization (TRO)} is vital for multi-hop networks.

These two scenarios correspond to optimization problems that share the same underlying structure, but differ in decision variable. Both one-hop and multi-hop settings can be modeled as Mixed Integer Nonlinear Programming (MINLP) problems~\cite{tehEnablingQuasiStaticReconfigurable2023,zhangGeminiPracticalReconfigurable2021}. In the multi-hop case, the joint Topology and Routing Optimization problem involves optimizing both topology and routing variables, leading to  approximately $\mathcal{O}(N^2)$ integer variables and $\mathcal{O}(N^3)$ continuous variables for a network comprising $N$ PoDs. For example, with $N = 128$, the problem includes more than 16 thousand integer and 2 million continuous variables, making it intractable for commercial solvers. In the one-hop case, the routing is fixed to direct paths, optimizing only the logical topology. This reduces the problem to $\mathcal{O}(N^2)$ integer variables and one continuous variable, but it remains complex due to combinatorial nature. More importantly, because one-hop optimization is typically used in latency-sensitive reconfiguration scenarios, it requires efficient algorithms that can deliver near-instantaneous solutions.

Prior work has addressed these challenges via relaxation and heuristic methods. COUDER~\cite{tehEnablingQuasiStaticReconfigurable2023} relaxes the joint problem into a linear program (LP) and rounds the solution to a feasible topology, but suffers from rounding errors due to LP relaxation and incurs high solver overhead, limiting scalability to large networks. TO-specific methods typically adopt either maximum-cost flow (MCF) formulations~\cite{zhaoMinimalRewiringEfficient} or Birkhoff–von Neumann (BvN) decomposition-based matching~\cite{liuSchedulingTechniquesHybrid2015,porterIntegratingMicrosecondCircuit2013}, which improve tractability but offer limited control over global metrics such as maximum link utilization (MLU). These limitations highlight the need for a unified, efficient, and flexible framework that supports both architectural paradigms.

Our key insight is that the core challenge in both TRO and TO lies in the combinatorial nature of topology decisions, where integer-valued link allocations must satisfy port and capacity constraints under given traffic demand.
To address this, we propose two complementary ideas.
First, we decompose TRO into two coordinated sub-problems: topology optimization (TO) and routing optimization (RO).
This decoupling enables an alternating optimization strategy, where TO is solved under fixed routing and RO under fixed topology.
Second, we observe that the TO subproblem exhibits a useful monotonicity: increasing the maximum link utilization (MLU) enlarges the feasible set of link allocations.
This property allows TO to be solved optimally via an efficient binary search, which significantly outperforms general-purpose solvers in runtime while preserving exact optimality.

Building on these insights, we propose \textbf{Alternating Topology and Routing Optimization (ATRO)}, a modular and scalable framework that alternates between TO and RO to progressively improve network performance.
ATRO guarantees monotonic improvement in MLU and supports warm-starting from any feasible initialization, including solutions produced by existing methods such as COUDER~\cite{tehEnablingQuasiStaticReconfigurable2023}.
To implement this framework, we develop the Accelerated Binary Search Method (ABSM) for TO and adopt Traffic Engineering (TE)~\cite{narayanan_solving_2021,alqiamTransferableNeuralWAN2024,Yang_Learning,mao2025fastsolverfreealgorithmtraffic,liu_figret_2023,xu_teal_2023,perry_dote_nodate}  accelerators for RO.
These components allow ATRO to efficiently support both one-hop and multi-hop scenarios; notably, when routing is performed using solver-free TE accelerators, ATRO operates entirely without invoking any commercial solver.

We extensively evaluate ATRO in both scenarios and find that it consistently outperforms previous methods in maximum link utilization (MLU) and runtime, achieving up to $10\times$ speed-ups while maintaining high solution quality. Its modular structure, scalability, and practical efficiency make it suitable for a wide range of deployment scenarios. To facilitate further research, our code is available at \cite{ymmao-xjtusiiYingmingMaoATRO2025}.

In summary, this paper makes three primary contributions:
\begin{itemize}
    \item We propose ATRO, a unified and modular framework for efficiently solving both topology-only and joint topology-routing optimization problems in reconfigurable DCNs.
    \item We develop ABSM, a fast and exact topology optimizer based on binary search, and integrate TE accelerators for scalable routing updates.
    \item We evaluate ATRO on both one-hop and multi-hop scenarios, demonstrating superior MLU and significantly faster runtime compared to prior methods.
\end{itemize}

\section{Problem Formulation and Key Insight}

\subsection{Problem Formulation}
\label{sec:Problem Formulation}

Similarly to~\cite{tehEnablingQuasiStaticReconfigurable2023,zhangGeminiPracticalReconfigurable2021}, we consider a reconfigurable DCN with $N$ P, modeled as a directed graph $G = (V, E)$, where $V$ is the set of Po and $E$ comprises all potential logical links. Each PoD $i$ has $R_i$ ports, each of capacity $S_i$, and the capacity of any logical link $(i,j)$ is $S_{i,j} = \min(S_i, S_j)$. The traffic demand from PoD $i$ to PoD $j$ is denoted $D_{i,j}$. Let $n_{i,j} \in \mathbb{Z}_{\ge 0}$ be the number of logical links from PoD $i$ to $j$ and each PoD $i \in V$ has $R_i$ ports. Similar to Google~\cite{poutievskiJupiterEvolvingTransforming2022}, we assume traffic is routed via at most two hops. We define a path as a triad $(s,k,d)$, where $s$ is the source, $d$ is the destination, and $k$ is either an intermediate relay (if $k \ne d$) or the destination itself (if directly routed). The set of all permissible paths is denoted $\mathcal{P}$, and for each source-destination pair $(i,j)$, we define the set of valid relay options as $\mathcal{K}_{ij} = \{k \mid (i,k,j) \in \mathcal{P} \}$. We introduce $f_{i,j,k}$ to denote the fraction of $i \to j$ traffic routed through node $k$ (with $k = j$ representing direct routing). Our goal is to jointly optimize the integer-valued logical topology ${n_{i,j}}$ and the continuous routing variables ${f_{i,j,k}}$ to minimize the maximum link utilization $u$. The joint topology and routing optimization (TRO) problem is formulated as \cref{eq:optimization_problem}.

Constraint~\eqref{eq:constraint1} limits the total number of logical links per PoD by available port count. Constraint~\eqref{eq:constraint2} enforces link symmetry, and requires that $n_{i,j} \in \mathbb{Z}_{\ge 0}$, i.e., each logical link count must be a non-negative integer. Constraints~\eqref{eq:constraint3}–\eqref{eq:constraint4} restrict routing to valid two-hop paths. Constraint~\eqref{eq:constraint5} bounds
link load by capacity. While our framework naturally generalizes to longer path-based routing~\cite{zhangGeminiPracticalReconfigurable2021}, we adopt the two-hop model for clarity and tractability.
\begin{subequations}
\label{eq:optimization_problem}
\begin{align}
    & \min_{f_{i,j,k},\ n_{i,j},\ u} \quad u 
     \label{eq:objective} \\
    \text{s.t.} \quad 
    & \sum_{\substack{j=1 , j\ne i}}^N n_{i,j} \le R_i, \quad \forall i, 
     \label{eq:constraint1} \\
    & n_{i,j} = n_{j,i}, \quad n_{i,j} \in \mathbb{Z}^+, \quad \forall i \ne j 
     \label{eq:constraint2}, \\
    & f_{i,j,k} \ge 0,\quad f_{i,j,k} = 0 \text{ if } k \notin \mathcal{K}_{ij},\quad \forall i,j,k 
     \label{eq:constraint3}, \\
    & \sum_{k \in \mathcal{K}_{ij}} f_{i,j,k} = 1, \quad \forall i \ne j 
     \label{eq:constraint4}, \\
    & \sum_{j'=1}^N f_{i,j',j} D_{i,j'} + \sum_{i'=1}^N f_{i',j,i} D_{i',j} 
    \le u  n_{i,j}  S_{i,j}, \quad \forall i \ne j.
     \label{eq:constraint5}
\end{align}
\end{subequations}

\noindent
\textbf{One-Hop Special Case.}
In the one-hop case, where all traffic is routed directly without intermediate relays, the routing variable $f_{i,j,k}$ becomes fixed (i.e., $f_{i,j,k} = \mathbf{1}_{k=j}
$). Under this setting, constraints~\eqref{eq:constraint3}–\eqref{eq:constraint5} collapse into a single per-link capacity constraint:
\[
D_{i,j} \le u \cdot n_{i,j} \cdot S_{i,j}, \quad \forall i \ne j.
\]
This yields a \textit{topology optimization (TO)} problem over $\{n_{i,j}\}$, which remains challenging due to its combinatorial nature.

\subsection{Limitations of Relaxation-Based Approaches}
\label{sec:Relaxation Limitations}

The TRO problem defined in~\eqref{eq:optimization_problem} is a mixed-integer nonlinear program (MINLP), due to the coupling of integer-valued topology variables $n_{i,j}$ and continuous routing variables $f_{i,j,k}$ via bilinear constraints (e.g., $u \cdot n_{i,j}$). A common approach, adopted by COUDER~\cite{tehEnablingQuasiStaticReconfigurable2023}, is to first linearize the formulation into a mixed-integer linear program (MILP), then relax it to an LP, and finally perform heuristic rounding to recover an integral topology.

However, this three-stage pipeline—MINLP $\rightarrow$ MILP $\rightarrow$ LP + rounding—has two key drawbacks. First, it fails to capture the combinatorial structure of topology selection, often producing fractional link allocations that cannot be feasibly rounded (Fig.\ref{fig:feasible_of_couder}). Second, even the relaxed LP becomes intractable at scale, containing millions of variables for large DCNs. These limitations also apply to simpler one-hop TO scenarios, where fast and reliable decisions are essential.

\begin{figure}[tbh]
    \centering
    \includegraphics[width=0.8\linewidth]{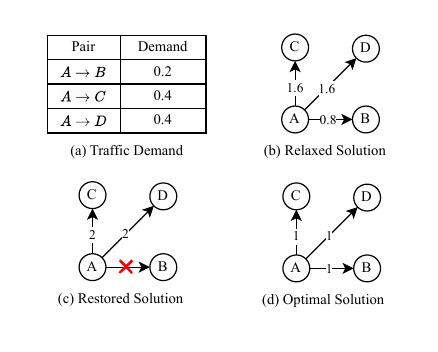}
    \caption{Illustration of infeasibility in the rounded TRO solution produced by COUDER. Assume each PoD has $R_i=4$ ports. Subfigure (b) shows a relaxed LP solution with fractional link allocations. Subfigure (c) demonstrates an infeasible rounding outcome where PoD A and B are disconnected. Subfigure (d) shows a feasible and optimal integer solution.}
    \label{fig:feasible_of_couder}
\end{figure}

\subsection{Key Insight and Theoretical Discussion}
\label{sec:atro_insight}

We propose ATRO, a scalable and fast framework for TRO, based on two core insights. 
First, the TRO problem can be decomposed into two interacting subproblems—topology optimization (TO) and routing optimization (RO)—that are coupled only through the capacity constraints~\cref{eq:constraint5}.
Specifically, given routing decisions $\{f_{i,j,k}\}$ and demands $D_{i,j}$, we define the total load on each logical link $(i,j)$ as:
\begin{equation}
T_{i,j} = \sum_{j'=1}^N f_{i,j',j} D_{i,j'} + \sum_{i'=1}^N f_{i',j,i} D_{i',j}.
\label{eq:linkload}
\end{equation}
This representation allows all flow-related constraints in the original TRO problem to be expressed through the per-link loads $\{T_{i,j}\}$, thereby decoupling the routing logic from the topology variables in the capacity constraint~\eqref{eq:constraint5}. Conversely, for fixed topologies $\{n_{i,j}\}$, the routing problem reduces to computing $\{f_{i,j,k}\}$ under fixed capacity constraints.  Second, the TO subproblem exhibits a monotonic feasibility structure: for fixed $\{T_{i,j}\}$, increasing the utilization threshold $u$ enlarges the feasible set of integer topologies. This property enables efficient binary search to identify the optimal $u$ without requiring general-purpose solvers.

\noindent
\textbf{Decomposition Strategy.}
We alternate between solving the following subproblems:

\textit{(i) Topology Optimization (TO):} Given fixed routing $\{f_{i,j,k}\}$ and corresponding $T_{i,j}$, solve:
\begin{equation}
\begin{aligned}
& \min_{n_{i,j},\ u} \quad u \\
\text{s.t.} \quad 
& \sum_{j \ne i} n_{i,j} \le R_i,\quad n_{i,j} = n_{j,i} \in \mathbb{Z}^+, \\
& T_{i,j} \le u \cdot n_{i,j} \cdot S_{i,j},\quad \forall i \ne j.
\end{aligned}
\label{eq:to_subproblem}
\end{equation}
We exploit monotonicity to perform binary search on $u$.

\textit{(ii) Routing Optimization (RO):} Given fixed topology $\{n_{i,j}\}$ and link capacities, solve:
\begin{equation}
\begin{aligned}
& \min_{f_{i,j,k},u} \quad u \\
\text{s.t.} \quad 
& f_{i,j,k} \ge 0,\quad \sum_{k \in \mathcal{K}_{ij}} f_{i,j,k} = 1,\quad \forall i \ne j, \\
    & \sum_{j'=1}^N f_{i,j',j} D_{i,j'} + \sum_{i'=1}^N f_{i',j,i} D_{i',j} 
    \le u  n_{i,j}  S_{i,j}, \quad \forall i \ne j. 
\end{aligned}
\label{eq:ro_subproblem}
\end{equation}
\cref{eq:ro_subproblem} is a standard traffic engineering problem that can be solved using LP solvers or accelerated using TE accelerator.

\label{sec:atro_advantage}
\noindent
\textbf{Theoretical Discussion.}
ATRO enjoys several desirable properties that stem from its decomposition-based design and alternating optimization structure. These include modularity, convergence guarantees, and robust performance.

\begin{itemize}
    \item \textbf{Modular decomposition.}  
    ATRO decouples topology and routing to simplify the joint optimization, allowing each subproblem to be solved independently via specialized techniques. This modularity naturally generalizes to multi-hop routing: while our TO component optimizes links based on aggregate loads, the RO component can seamlessly integrate any standard multi-hop TE algorithm.

    \item \textbf{Monotonic objective descent.}  
    Each TO update reduces the utilization threshold $u$ (or leaves it unchanged), and the subsequent RO step recomputes routing under the updated topology. This guarantees non-increasing objective values and convergence to a fixed point.

    \item \textbf{Synergy and Robustness.} ATRO converges to near-optimal solutions through the synergy of its components: TO reshapes the global topology while RO mitigates local congestion. In rare sparse-demand cases, ATRO may retain links that an exact solver would prune. This occurs because the routing step lacks an incentive to drive specific link loads to absolute zero, preventing the topology step from removing them. While slightly suboptimal in link count, this behavior enhances robustness by maintaining additional path diversity.

\end{itemize}

We next describe the algorithmic design of ATRO and how the alternating updates are coordinated.

\section{ATRO Design}
\label{sec:ATRO}

Building on the decomposition strategy introduced in~\S\ref{sec:atro_insight}, we now present the design of \emph{Alternating Topology and Routing Optimization (ATRO)}, a scalable and fast framework that alternately solves the topology and routing subproblems.

\subsection{Overview}
\label{sec:atro-overview}
ATRO consists of three coordinated components: \textit{TO}, \textit{RO}, and a \textit{Refinement} module. The \textit{TO Component} computes a logical topology minimizing MLU using an Accelerated Binary Search Method (ABSM). The \textit{RO Component} then determines routing decisions using general traffic engineering techniques. Between these two, the \textit{Refinement Component} opportunistically reallocates unused ports to alleviate link overloads, improving routing flexibility without increasing MLU. These modules form an iterative loop: routing informs topology updates and refinement, which in turn guide routing. For one-hop cases, ATRO simplifies to a single \textit{TO} invocation. Fig.~\ref{fig:atro-overview} illustrates the architecture.

\begin{figure}[t]
    \centering
\includegraphics[width=0.8\linewidth]{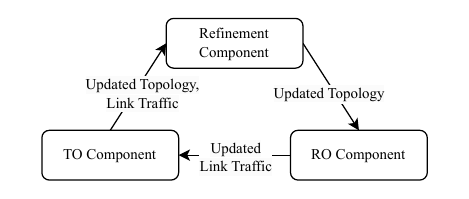}
    \caption{
    Overview of ATRO’s architecture. The framework iteratively coordinates three components—\textit{TO Component}, \textit{Refinement Component}, and \textit{RO Component}—through shared intermediate variables.
    }
    \label{fig:atro-overview}
\end{figure}

\subsection{TO Component}
\label{sec:to-component}

The Topology Optimization (TO) subproblem in ATRO is responsible for determining the logical topology $\{n_{i,j}\}$ that minimizes the MLU $u$, given the current traffic load $T_{i,j}$ and the port limits $R_i$. Formally, TO aims to find the minimum $u$ such that the total traffic $T_{i,j}$ can be delivered over a topology that satisfies port constraints and supports link capacities $S_{i,j}$. Since link counts $n_{i,j}$ are required to be integers, this is a combinatorial problem.

Our approach is based on the observation that for a given target utilization level $u > 0$, we can efficiently test whether a feasible topology exists by constructing a minimal link allocation that supports the traffic under that $u$. This leads to an efficient binary search scheme for solving TO.

We first formalize the feasibility condition that serves as the backbone of our algorithm in Theorem~\ref{thm:theorem1}.

\begin{theorem}[Feasibility of TO for fixed $u$]
\label{thm:theorem1}
Let $u > 0$ be a candidate MLU. Define the number of logical links required between each PoD pair $(i,j)$ as:
\begin{equation}
n^{(U)}_{i,j} = n^{(U)}_{j,i} = \left\lceil \max \left\{ \frac{T_{i,j}}{u \cdot S_{i,j}}, \frac{T_{j,i}}{u \cdot S_{j,i}} \right\} \right\rceil, \quad \forall i \ne j.
\label{eq:n_U_def}
\end{equation}
Then $u$ is a feasible solution to the TO subproblem if and only if
\begin{equation}
\sum_{j \ne i} n^{(U)}_{i,j} \le R_i, \quad \forall i.
\label{eq:feasibility_condition}
\end{equation}
\end{theorem}

\begin{proof}
We prove both sufficiency and necessity.

\textbf{Sufficiency.} According to \eqref{eq:n_U_def}, we know:
\begin{equation}
n^{(U)}_{i,j} \ge \frac{T_{i,j}}{u \cdot S_{i,j}}, \quad \forall i \ne j,
\end{equation}
which implies:
\begin{equation}
u \cdot n^{(U)}_{i,j} \cdot S_{i,j} \ge T_{i,j}.
\label{eq:sufficiency_link_utilization}
\end{equation}
We consider the allocation \( n_{i,j} = n^{(U)}_{i,j} \). This yields a total capacity of \( n_{i,j} \cdot S_{i,j} \) between PoD \( i \) and \( j \), which, by \eqref{eq:sufficiency_link_utilization}, suffices to carry the traffic \( T_{i,j} \). If the port constraint in \eqref{eq:feasibility_condition} also holds, then the total number of links at each PoD \( i \) does not exceed its limit \( R_i \), and the resulting topology satisfies all constraints in \eqref{eq:to_subproblem}, making it feasible under \( u \).

\textbf{Necessity.} If $u$ is a feasible solution to the TO subproblem~\eqref{eq:to_subproblem} we know that there exists a group of $n_{i, j}$ such that all constraints in \eqref{eq:to_subproblem} are satisfied by $(u, n_{i, j})$. Therefore we have:
\begin{equation}
n_{i,j} \ge \frac{T_{i,j}}{u \cdot S_{i,j}}, \quad \text{and similarly} \quad n_{j,i} \ge \frac{T_{j,i}}{u \cdot S_{j,i}}.
\end{equation}
Due to the symmetry constraint $n_{i,j} = n_{j,i}$, we have:
\begin{equation}
n_{i,j} \ge \left\lceil \max \left\{ \frac{T_{i,j}}{u \cdot S_{i,j}}, \frac{T_{j,i}}{u \cdot S_{j,i}} \right\} \right\rceil = n^{(U)}_{i,j}.
\end{equation}
Therefore,
\begin{equation}
\sum_{j \ne i} n^{(U)}_{i,j} \le \sum_{j \ne i} n_{i,j} \le R_i,
\end{equation}
which proves the necessity of \eqref{eq:feasibility_condition}.
\end{proof}

Theorem~\ref{thm:theorem1} result enables closed-form feasibility checking for any given $u$ and yields the minimal link allocation $n^{(U)}_{i,j}$ satisfying both capacity and port constraints. Since feasibility is monotonic in $u$—i.e., if $u$ is feasible, any $u' > u$ is also feasible—binary search can efficiently find the minimum feasible value. Furthermore, each successful test returns a topology that allows us to tighten the upper bound: given $n^{(U)}_{i,j}$ from \eqref{eq:n_U_def}, the smallest $u$ supporting the current allocation satisfies:
\begin{equation}
u \ge \frac{T_{i,j}}{n^{(U)}_{i,j} \cdot S_{i,j}}, \quad \forall i \ne j.
\end{equation}
This yields a refined upper bound on feasible utilization:
\begin{equation}
\tilde{u} = \max_{i \ne j} \left\{ \frac{T_{i,j}}{n^{(U)}_{i,j} \cdot S_{i,j}} \right\}.
\label{eq:util_bound_refine}
\end{equation}

\begin{theorem}
\label{thm:absm_feasibility}
If $u$ is feasible for the TO subproblem, then $\tilde{u}$ computed via \eqref{eq:util_bound_refine} is also feasible.
\end{theorem}

\begin{proof}
By definition of $n^{(U)}_{i,j}$ and the sufficiency proof of Theorem~\ref{thm:theorem1}, we have that if $u$ is feasible, then $n^{(U)}_{i,j}$ supports all traffic. The refinement in \eqref{eq:util_bound_refine} ensures that no link utilization exceeds $\tilde{u}$. Therefore, the same topology $n^{(U)}_{i,j}$ remains feasible under $\tilde{u}$.
\end{proof}

According to \Cref{thm:absm_feasibility}, we solve the TO subproblem via a binary search over the utilization threshold $u$, and accelerate convergence by refining the upper bound using \eqref{eq:util_bound_refine} whenever feasibility is confirmed. To initialize the search interval, we start from $u = 1$ and increment it step by step (e.g., by 1) until the feasibility condition in \Cref{thm:theorem1} is satisfied; the first feasible value is recorded as the initial upper bound $M$. This adaptive initialization avoids manual tuning and typically completes in under 5 iterations, with $M$ usually no greater than 5. Once the interval $[\underline{u}, \overline{u}] = [0, M]$ is established, we iteratively bisect it to test feasibility at the midpoint. If feasible, the upper bound is tightened to the implied minimum utilization $\tilde{u}$ computed from the current link allocation; otherwise, the lower bound is raised. This dynamic tightening is what accelerates the binary search. The full process is outlined in Algorithm~\ref{alg:absm}.


\begin{algorithm}[t]
\caption{Accelerated Binary Search Method (ABSM)}
\label{alg:absm}
\begin{algorithmic}[1]
\Require Traffic matrix $T_{i,j}$, link capacities $S_{i,j}$, port limits $R_i$, threshold $\varepsilon$
\Ensure Optimal link allocation $n_{i,j}$ and MLU $u$
\State Initialize bounds: $\underline{u} \gets 0$, $\overline{u} \gets M$
\While{$\overline{u} - \underline{u} > \varepsilon$}
    \State $u \gets (\underline{u} + \overline{u}) / 2$
    \State Compute $n^{(U)}_{i,j}$ using Eq.~\eqref{eq:n_U_def}
    \If{$\sum_{j \ne i} n^{(U)}_{i,j} \le R_i \quad \forall i$}
        \State $\tilde{u} \gets \max_{i \ne j} \frac{T_{i,j}}{n^{(U)}_{i,j} \cdot S_{i,j}}$
        \State $\overline{u} \gets \tilde{u}$
    \Else
        \State $\underline{u} \gets u$
    \EndIf
\EndWhile
\State \Return $\overline{u}, \{n^{(U)}_{i,j}\}$
\end{algorithmic}
\end{algorithm}

\noindent
\textbf{Complexity.}  
Given the upper bound $M$ and convergence threshold $\varepsilon$, ABSM performs at most $\mathcal{O}(\log_2(M/\varepsilon))$ iterations. Each iteration consists of simple element-wise operations—vectorized division, max, and summation—over $N \times N$ matrices. These operations are highly parallelizable and execute in near-constant time on modern hardware. Thus, the total runtime is dominated by $\mathcal{O}(\log(M/\varepsilon))$ lightweight steps, ensuring excellent scalability even for large DCN.

\subsection{RO Component}
\label{sec:ro-component}

The routing optimization (RO) subproblem is a standard traffic engineering (TE) task under fixed topology. It can be directly solved via linear programming (LP); however, LP solvers often incur high computational cost at large scale. To enhance scalability, we adopt recent TE accelerators that substantially reduce runtime without sacrificing solution quality.

\noindent
\textbf{Learning-based TE accelerators.}
Deep learning methods such as Figret~\cite{liu_figret_2023} and DOTE~\cite{perry_dote_nodate} achieve fast inference but are typically limited to fixed logical topologies. More recent techniques, including FNC\cite{Yang_Learning}, RedTE~\cite{redte}, and HARP~\cite{alqiamTransferableNeuralWAN2024}, have extended this capability to partially variable topologies, offering potential integration into hybrid systems.

\noindent
\textbf{Decomposition-based TE accelerators.}
These approaches partition the TE problem into tractable subproblems that can be solved efficiently in sequence or parallel. Notable examples include POP~\cite{narayanan_solving_2021}, LP-top~\cite{xu_teal_2023}, and SSDO~\cite{mao2025fastsolverfreealgorithmtraffic}. Among them, SSDO offers a compelling balance of speed and quality. It is solver-free, supports warm-start from previous solutions, and delivers near-optimal results, making it well-suited for integration into ATRO.

\noindent
\textbf{ATRO Integration.}
While ATRO is compatible with any TE algorithm, we adopt SSDO in all subsequent experiments due to its efficiency and natural alignment with our alternating framework. The RO module remains modular and can readily integrate more advanced or domain-specific TE methods as they become available.

\subsection{Refinement Component}
\label{sec:refinement}
ATRO converges to a fixed-point solution that satisfies two conditions: (1) given the current topology, no routing update can further reduce the MLU, and (2) given the current routing, no topology adjustment can further reduce the MLU without violating port constraints. While such a solution is both feasible and stable, it may still be suboptimal. This is because the topology optimization (TO) subproblem often admits multiple  distinct optimal solutions that yield the same MLU. However, some of these topologies may severely constrain the routing optimization (RO) subproblem by limiting path diversity and underutilizing available network capacity, thereby reducing its ability to distribute traffic effectively.

To address this, we introduce a lightweight refinement mechanism that selectively increases routing flexibility. The key idea is to identify high-load logical links and augment them with additional connections where possible, provided both endpoints have idle ports. Fig.~\ref{fig:add-link} illustrates this in a 3-PoD example. In the initial topology (Fig.~\ref{fig:add-link-b}), PoD A connects to both B and C with one link each and routes traffic directly. Due to port limits, only one extra link can be added, and neither A--B nor A--C improves the MLU under fixed routing, leading to a local optimum. The refinement step adds links between A--C and B--C (Fig.~\ref{fig:add-link-c}), which unlocks better routing options (Fig.~\ref{fig:add-link-d}) by relaying through PoD C, reducing MLU from 0.2 to 0.15. This strategy preserves feasibility while expanding the solution space for subsequent routing updates.

The refinement procedure is detailed in Algorithm~\ref{alg:topo-refine}. It computes the current link utilization, ranks PoD pairs by descending utilization, and iteratively adds links where both sides have residual port capacity. Importantly, this process does not alter existing traffic allocations or exceed port limits. It is applied immediately after each TO update, preserving feasibility while strategically improving routing flexibility. In practice, this helps ATRO escape local optima and achieve more balanced traffic distributions.
\begin{figure}[thb]
    \centering
    \begin{subfigure}[b]{0.4\linewidth}
        \includegraphics[width=\linewidth]{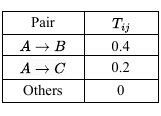}
        \caption{Traffic Demand}
        \label{fig:add-link-a}
    \end{subfigure}
    \hfill
    \begin{subfigure}[b]{0.4\linewidth}
        \includegraphics[width=\linewidth]{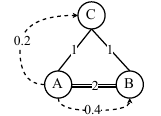}
        \caption{Local Optima}
        \label{fig:add-link-b}
    \end{subfigure}

    \begin{subfigure}[b]{0.4\linewidth}
        \includegraphics[width=\linewidth]{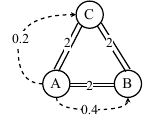}
        \caption{Topology Refinement}
        \label{fig:add-link-c}
    \end{subfigure}
    \hfill
    \begin{subfigure}[b]{0.4\linewidth}
        \includegraphics[width=\linewidth]{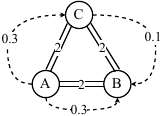}
        \caption{New Optima}
        \label{fig:add-link-d}
    \end{subfigure}
    
    \caption{Illustration of  topology refinement in a 3-PoD example. 
Each PoD has a port budget of 4, and each logical link has unit capacity. 
Solid lines represent bidirectional logical links, annotated with the number of connections. 
Dashed arrows represent traffic loads, annotated with flow volumes.}

    \label{fig:add-link}
\end{figure}

\begin{algorithm}[tbh]
\caption{Traffic-Aware Topology Refinement (Refine)}
\label{alg:topo-refine}
\begin{algorithmic}[1]
\Require Topology matrix $n_{i,j}$, Traffic load $T_{i,j}$, link capacities $S_{i,j}$, port limits $R_i$
\Ensure Refined topology $n_{i,j}$
\State Compute used ports: $R_c[i] \gets \sum_{j \ne i} n_{i,j}$
\State Compute remaining ports: $R_{\text{remain}}[i] \gets R_i - R_c[i]$
\State Compute utilization: $u_{i,j} \gets T_{i,j} / S_{i,j}$
\State Let $\texttt{SortedPairs} \gets$ PoD pairs $(i,j)$ sorted by $u_{i,j}$ in descending order
\For{each $(i,j) \in \texttt{SortedPairs}$}
    \If{$i = j$ or $R_{\text{remain}}[i] = 0$ or $R_{\text{remain}}[j] = 0$}
        \State \textbf{continue}
    \EndIf
    \State $\Delta \gets \min(R_{\text{remain}}[i], R_{\text{remain}}[j])$
    \State $n_{i,j} \gets n_{i,j} + \Delta$, $n_{j,i} \gets n_{j,i} + \Delta$
    \State $R_{\text{remain}}[i] \gets R_{\text{remain}}[i] - \Delta$, $R_{\text{remain}}[j] \gets R_{\text{remain}}[j] - \Delta$
\EndFor
\State \Return $n_{i,j}$
\end{algorithmic}
\end{algorithm}

\subsection{ATRO Algorithm Summary} \label{sec:atro-summary}

ATRO proceeds in an iterative loop with three modular components—\textit{TO}, \textit{Refinement}, and \textit{RO}—each addressing a distinct aspect of the joint optimization. Starting from an initial traffic load estimate, the TO component solves for a feasible topology; the Refinement step reallocates unused ports to improve balance; and the RO component computes optimal routing under the current topology. The procedure repeats until convergence, with each step maintaining feasibility and improving MLU.

ATRO presents numerous practical advantages. \begin{itemize}
    \item \textbf{Any-time usability.} Every intermediate solution remains feasible, allowing early termination with valid topology and routing, adaptable to varying run-time budgets.

    \item \textbf{Warm-start and hybrid integration.} ATRO supports initialization from arbitrary feasible states, and can be combined with external heuristics or prior solutions enabling efficient online re-optimization and collaborative algorithm design.

    \item \textbf{Solver-free and lightweight.}  
    By decomposing the TRO into combinatorial TO and continuous RO, ATRO eliminates the need for general-purpose solvers. TO is solved via binary search with closed-form checks, and with solver-free TE methods, RO can be executed without solvers, enabling full-pipeline scalability.

\end{itemize}

\noindent
\textbf{Special Case: One-Hop Topology Optimization.}  
In one-hop scenarios where routing is fixed to direct paths, only the TO subproblem is relevant. In this setting, ATRO reduces to a single call to the TO Component (i.e., ABSM), which efficiently computes the optimal one-hop topology via binary search. This highlights ATRO's versatility: it serves both as a high-speed optimizer for one-hop reconfiguration and as a general framework for multi-hop DCNs.

While ATRO supports arbitrary feasible initializations for routing, we empirically find that using direct-path routing consistently leads to faster convergence and better final performance. Unless otherwise specified, all evaluations in this paper adopt direct-path initialization by default. The complete procedure is summarized in Algorithm~\ref{alg:atro}.

\begin{algorithm}[tbh]
\caption{ATRO: Alternating Topology and Routing Optimization}
\label{alg:atro}
\begin{algorithmic}[1]
\Require Traffic demand $D_{i,j}$, port limits $R_i$, link capacity $S_i$
\Ensure Topology $n_{i,j}$, routing $f_{i,j,k}$, and maximum utilization $u$
\State\textbf{Initialize routing:} any feasible $f^{(0)}_{i,j,k}$ is allowed
\Statex \quad \textit{(e.g., direct path routing: $f^{(0)}_{i,j,k} \leftarrow 1$ if $k = j$, else $0$, for $i \ne j$)}
\State \textbf{Compute initial traffic:} 
\Statex \quad $T^{(0)}_{i,j} \gets \sum_{j'} f^{(0)}_{i,j',j} \cdot D_{i,j'} + \sum_{i'} f^{(0)}_{i',j,i} \cdot D_{i',j}$
\State \textbf{Initialize utilization:} $u^{(0)} \gets +\infty$, \quad $t \gets 0$

\Repeat
    \State \textbf{TO Component:} 
    \Statex \quad $n^{(t+1)}_{i,j} \gets \texttt{ABSM}(T^{(t)}_{i,j}, R_i, S_{i,j})$
    
    \State \textbf{Refinement Component:}
    \Statex \quad $n^{(t+1)}_{i,j} \gets \texttt{Refine}(n^{(t+1)}_{i,j}, T^{(t)}_{i,j}, S_{i,j}, R_i)$
    
    \State \textbf{RO Component:} 
    \Statex \quad $(f^{(t+1)}_{i,j,k}, u^{(t+1)}) \gets \texttt{SSDO}(D_{i,j}, n^{(t+1)}_{i,j}, S_{i,j})$
    
    \State \textbf{Traffic Update:}
    \Statex \quad $T^{(t+1)}_{i,j} \gets \sum_{j'} f^{(t+1)}_{i,j',j} \cdot D_{i,j'} + \sum_{i'} f^{(t+1)}_{i',j,i} \cdot D_{i',j}$
    \State $t \gets t+1$
\Until{$|u^{(t)} - u^{(t-1)}| < \varepsilon$}
\end{algorithmic}
\end{algorithm}

\section{Numerical Test}

We evaluate ATRO using the experimental setup described in \S\ref{sec:methodology}, focusing on two representative scenarios: one-hop settings, where only direct path are considered (\S\ref{sec:one_hop}) , and multi-hop  settings, which require joint optimization of topology and routing over longer timescales (\S\ref{sec:muti_hop}). For both cases, we compare ATRO against solver-based and heuristic baselines in terms of solution quality (MLU) and computation time. We further analyze its convergence behavior (\S\ref{sec:convergence}), warm-start capabilities, and the contribution of components through ablation studies (\S\ref{sec:Ablation}).

\subsection{Methodology} \label{sec:methodology}

\noindent\textbf{Topologies.} 
We evaluate ATRO on two categories of topologies: Meta’s production DCNs~\cite{roy_inside_2015}, including PoD and ToR levels, and four synthetically generated full-mesh topologies. As in prior studies~\cite{han2025highlyscalablellmclusters,tehEnablingQuasiStaticReconfigurable2023}, all logical links are assumed to have the same fixed capacity across the network during evaluation. Summary statistics are shown in Table~\ref{tab:NETWORK TOPOLOGIES}.

\begin{figure}[t]
    \centering
    \includegraphics[width=0.85\linewidth]{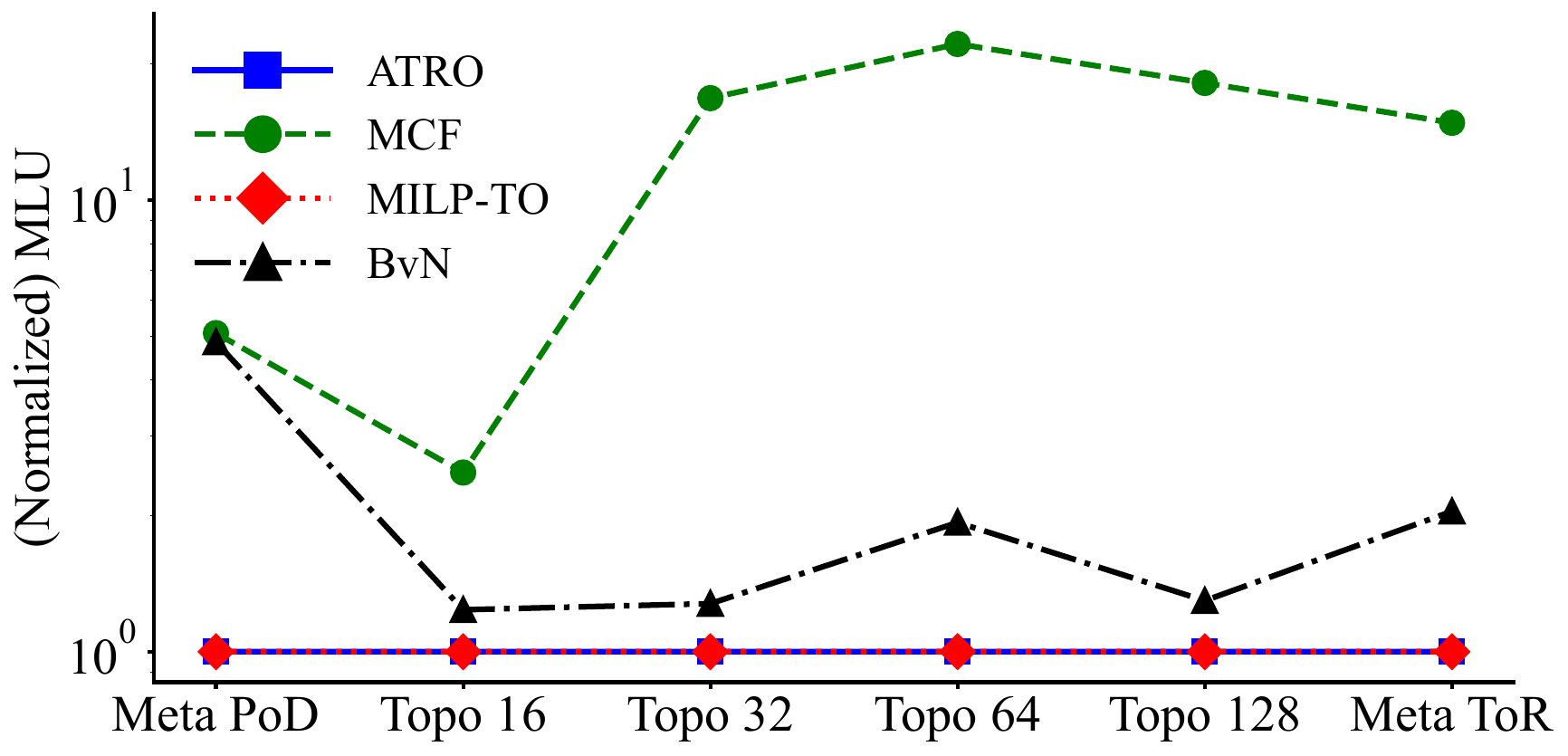}
    \caption{Average normalized MLU of ATRO and baselines on one-hop setting.}
    \label{fig:TE Quality single}
\end{figure}

\begin{table}[ht]
\caption{Network topologies used in evaluation.}
\centering
\begin{tabular}{ccccc}
\toprule
Type & Name & Nodes & Edges & Ports \\
\midrule
\multirow{2}{*}{Meta} & Meta PoD & 4 & 12 & 16 \\
                      & Meta ToR & 155 & 23870 & 256 \\
\midrule
\multirow{4}{*}{Synthetic} & Topo 16 & 16 & 240 & 32 \\
                           & Topo 32 & 32 & 992 & 64 \\
                           & Topo 64 & 64 & 4032 & 128 \\
                           & Topo 128 & 128 & 16256 & 256 \\
\bottomrule
\end{tabular}
\label{tab:NETWORK TOPOLOGIES}
\end{table}

\noindent\textbf{Traffic.}
For Meta topologies, we use a public production trace~\cite{roy_inside_2015}, aggregated at 1-second (PoD) and 100-second (ToR) intervals. For synthetic topologies, we combine AI traffic from \textit{RapidAISim} in~\cite{han2025highlyscalablellmclusters} with gravity model-based background flows~\cite{applegate_making_nodate,roughan_experience_2003}.

\noindent\textbf{Scenarios and Baselines.}
We evaluate ATRO in both one-hop and multi-hop settings, comparing it against representative solver-based and heuristic methods:

\begin{itemize}
\item \textbf{MILP-TO (one-hop):} An oracle baseline for the one-hop TO subproblem~\eqref{eq:to_subproblem} using commercial solver.   
\item \textbf{MCF (one-hop):} A fast approximation using minimum-cost flow models~\cite{zhaoMinimalRewiringEfficient}, where logical links are treated as unit flows. The resulting solution may violate symmetry constraints and is post-processed by downscaling asymmetric allocations to their minimum.
\item \textbf{BvN (one-hop):} A heuristic based on Birkhoff--von Neumann decomposition and disjoint perfect matchings~\cite{liuSchedulingTechniquesHybrid2015}. Like MCF, it does not enforce symmetry and applies a similar minimum-based adjustment.
\item \textbf{MILP (multi-hop):} Directly solves linearized TRO  using Gurobi, achieving optimal MLU but poor scalability.
\item \textbf{COUDER (multi-hop):} A two-stage relaxation-based method that solves a relaxed LP and heuristically reconstructs integer topologies~\cite{tehEnablingQuasiStaticReconfigurable2023}.
\end{itemize}

\noindent\textbf{Implementation.}
Algorithms are implemented in Python 3.8 using Gurobi 9.5.1 for solver-based baselines. Evaluations are conducted on an AMD EPYC 9654 with 384 GB RAM.

\subsection{One-Hop Evaluation (TO Component)}

\begin{figure}[t]
    \centering
    \includegraphics[width=0.85\linewidth]{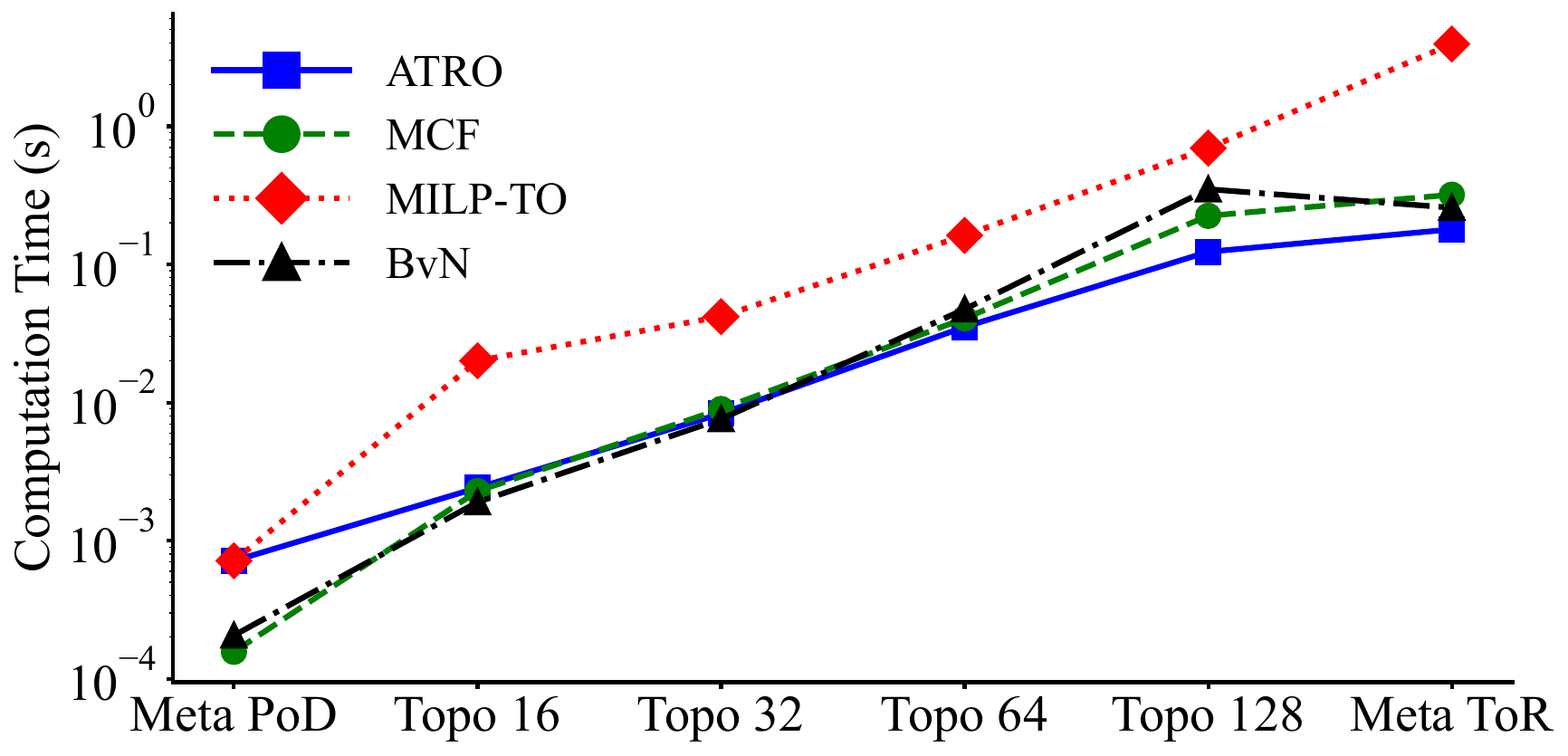}
    \caption{Average computation time of ATRO and baselines on one-hop setting.}
    \label{fig:TE solver time single}
\end{figure}
\label{sec:one_hop}
\noindent
\textbf{Performance Comparison.}
We evaluate ATRO against MILP-TO, MCF, and BvN in one-hop settings, where only the logical topology is optimized. As shown in Fig.~\ref{fig:TE Quality single}, ATRO (which reduces to ABSM in one-hop setting) achieves optimal MLU in all topologies. In contrast, MCF exhibits high variance and degrades significantly with increasing network size, while BvN performs reasonably on sparse cases (e.g., Meta PoD) but deteriorates rapidly in denser topologies. Fig.~\ref{fig:TE solver time single} further shows that ATRO maintains runtime under 100 ms, even on Meta ToR. MILP-TO is increasingly slower, taking over 1 seconds on Meta ToR, while BvN and MCF are fast but offer lower and unstable quality. Overall, ATRO dominates the tradeoff between solution quality and computational efficiency.

\noindent
\textbf{Stress Test.}
ATRO outperforms MILP-TO in one-hop TO problems, highlighting the need for speed in real-time topology reconfiguration despite MILP-TO’s capability to handle moderate sizes quickly. Stress testing with larger topologies (see \cref{tab:EXTENDED_TOPOLOGIES}) shows ATRO maintains optimal MLU (see Fig.\ref{fig:onehop_large_comparison}(a)) and near-second runtimes (see Fig.\ref{fig:onehop_large_comparison}(b)), unlike MILP-TO, which struggles with large-scale topology. These findings demonstrate ATRO's scalability and suitability for latency-sensitive scenarios in one hop scenario.

\begin{table}[ht]
\centering
\caption{Oversized topologies for stress testing.}
\begin{tabular}{ccccc}
\toprule
Type & Name & Nodes & Edges & Ports \\
\midrule
\multirow{2}{*}{Synthetic} & Topo 256 & 256 & 65280 & 512 \\
                          & Topo 512 & 512 & 261632 & 1024 \\
\bottomrule
\end{tabular}
\label{tab:EXTENDED_TOPOLOGIES}
\end{table}

\begin{figure}
    \centering
    \includegraphics[width=1\linewidth]{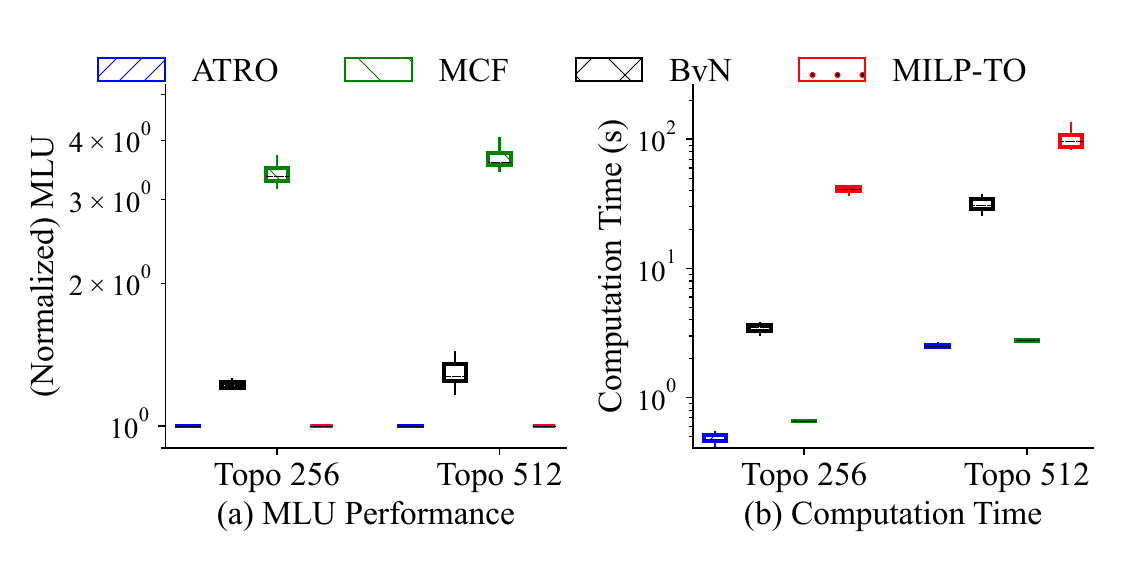}
    \caption{Comparison of ATRO and baselines on oversized topologies: (a) MLU; (b) Computation time.}
    \label{fig:onehop_large_comparison}
\end{figure}

\noindent
\textbf{Link Count Analysis.}
In addition to MLU and runtime, we evaluate the number of logical links each method provisions to satisfy traffic demand. As seen in Fig.~\ref{fig:linkcount}, ATRO consistently provisions the fewest links, thanks to \cref{thm:theorem1}. It minimizes the consumption of limited OCS port resources, leaving more "free" ports to accommodate sudden traffic bursts or background flows. MILP-TO often over-provisions due to solver rounding effects, leading to inefficiencies. MCF and BvN lack explicit link minimization objectives and hence result in inflated configurations, particularly in larger topologies like Topo 64 and 128.

\begin{figure}
    \centering
    \includegraphics[width=0.85\linewidth]{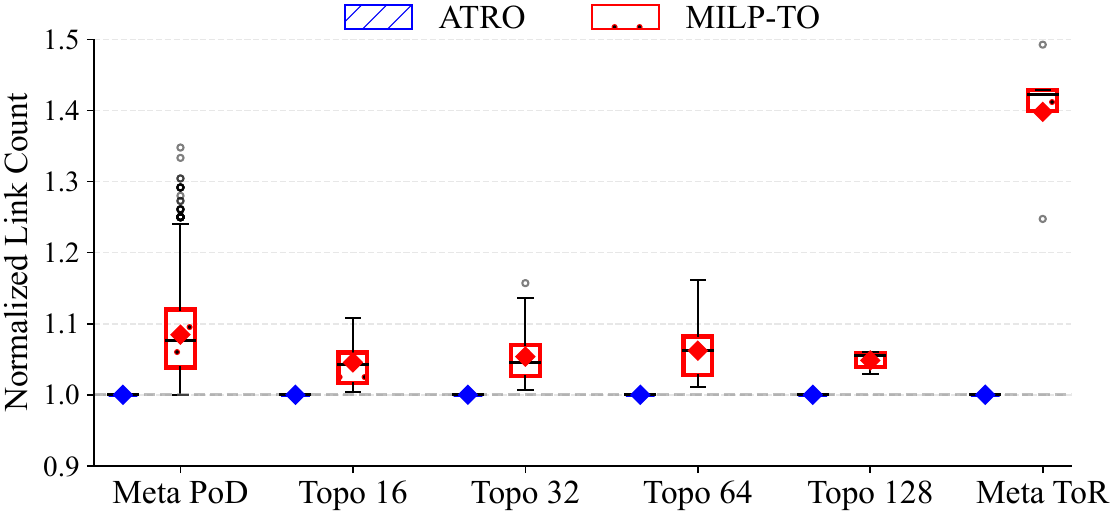}
    \caption{The count of logical links in practical topologies normalized by ATRO (lower is better).}
    \label{fig:linkcount}
\end{figure}

\noindent
\textbf{Summary.}
ATRO offers optimal MLU and compact logical topologies, while maintaining low computation time and scaling effectively to networks with hundreds of nodes. It consistently outperforms MILP and heuristic baselines.

\subsection{Multi-Hop Evaluation (TRO Problem)} 
\label{sec:muti_hop}

We evaluate the full ATRO framework under multi-hop settings, where both logical topology and routing must be jointly optimized. Fig.\ref{fig:TE Quality} and Fig.~\ref{fig:TE solver time} report normalized MLU and average computation time across six topologies. The analysis for each baseline is as follows:

\begin{figure}
    \centering
    \includegraphics[width=0.85\linewidth]{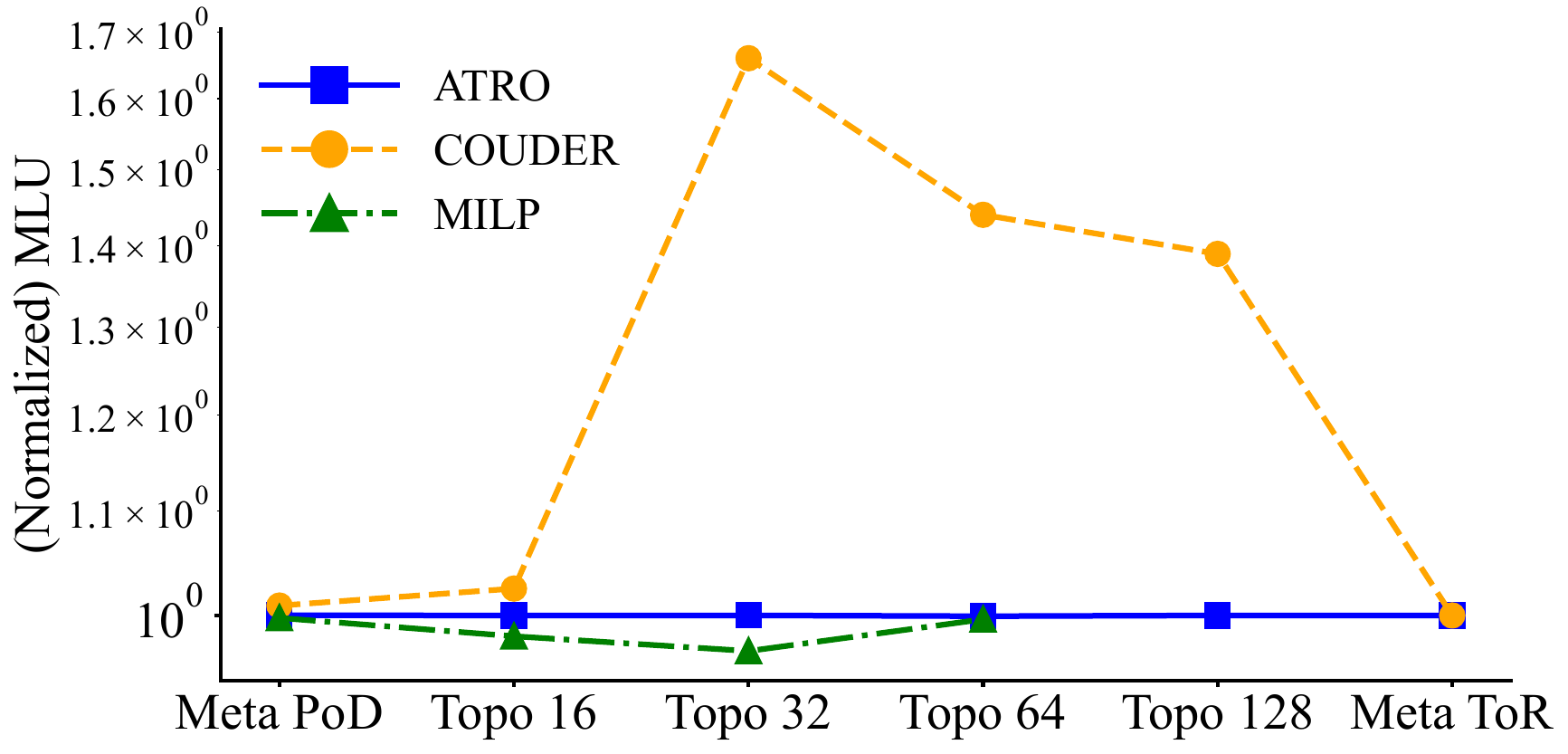}
    \caption{Average normalized MLU of ATRO and baselines in multi-hop setting.}
    \label{fig:TE Quality}
\end{figure}

\begin{figure}
    \centering
    \includegraphics[width=0.85\linewidth]{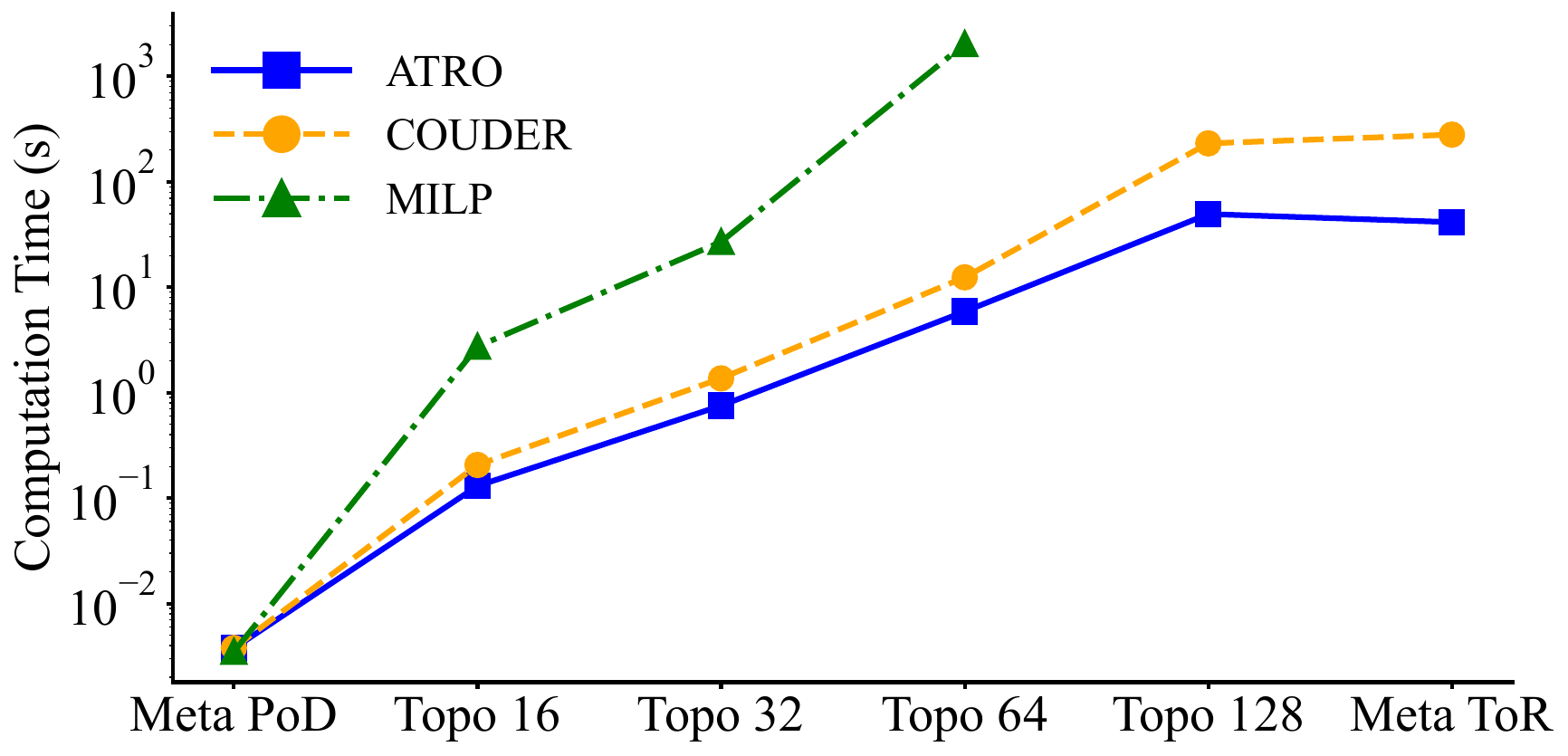}
    \caption{Average computation time of ATRO and baselines in multi-hop setting.}
    \label{fig:TE solver time}
\end{figure}

\noindent
\textbf{MILP (oracle baseline):} 
MILP represents the theoretical optimum by solving the full TRO problem via Gurobi. As shown in Fig.~\ref{fig:TE Quality}, it delivers the best MLU on small topologies like Meta PoD and Topo 16. However, Fig.~\ref{fig:TE solver time} reveals its critical weakness: computation time increases exponentially with scale, becoming intractable beyond Topo 64. On Topo 128 and Meta ToR, it fails to complete within 20,000 seconds.

\noindent
\textbf{COUDER:} 
 COUDER performs reasonably well on small networks but exhibits clear degradation on mid-sized topologies. On Topo 32 and Topo 64, COUDER yields over 1.5\texttimes{} higher MLU than ATRO. An exception occurs on Meta ToR, where its MLU appears competitive. This is due to the extremely skewed traffic pattern in Meta ToR, which admits many near-optimal configurations, allowing even COUDER's coarse topology recovery to perform well. Even so, as seen in Fig.~\ref{fig:TE solver time}, its runtime remains significantly higher than ATRO.

\noindent
\textbf{ATRO:} 
ATRO consistently achieves low MLU while maintaining high efficiency. In Fig.~\ref{fig:TE Quality}, it matches MILP on small topologies and outperforms COUDER significantly on larger ones. In Fig.~\ref{fig:TE solver time}, ATRO's runtime grows gradually, remaining practical even on the largest evaluated topology. On Topo 128, it is up to 5\texttimes{} faster than COUDER. Unlike MILP and COUDER, ATRO can avoid commercial solvers and remains robust across varying traffic characteristics. On Topo 32, ATRO exhibits a slight MLU gap relative to MILP, attributable to the extremely sparse demand matrix: many source-destination pairs require no connectivity. While MILP can prune such links entirely, ATRO’s alternating structure tends to preserve minimal connectivity, as the RO step rarely drives link utilization to zero. This leads to mild suboptimality but ensures solution feasibility and topological stability.

\subsection{Analysis of Convergence Process} 
\label{sec:convergence}

We analyze the convergence behavior of ATRO by tracking normalized MLU over iterations on representative samples from Meta PoD and Topo 16, as shown in Figure~\ref{fig:convergence}. The initial point is obtained by applying the TO component to the input traffic, which often provides a strong starting topology.

ATRO exhibits consistently monotonic improvement, with MLU decreasing at each iteration, as guaranteed in \S\ref{sec:atro_insight}. Most samples converge within one or two rounds, and subsequent iterations yield diminishing improvements—indicating that early stopping is often sufficient. Since both TO and RO components are lightweight (see \S\ref{sec:one_hop}), ATRO can deliver high-quality solutions with minimal delay.

To further quantify convergence efficiency, we measure the number of iterations required for convergence across all samples in six topologies. Figure~\ref{fig:convergence_round_stats} shows that for small to medium-sized topologies (Meta PoD, Topo 16/32/64), over 95\% of samples converge within two iterations. Even in large-scale topologies like Topo 128 and Meta ToR, the majority of cases require no more than three iterations. This empirical evidence confirms ATRO’s scalability and suitability for both low-latency and large-scale deployment.

\begin{figure}[tbh]
    \centering
    \includegraphics[width=0.85\linewidth]{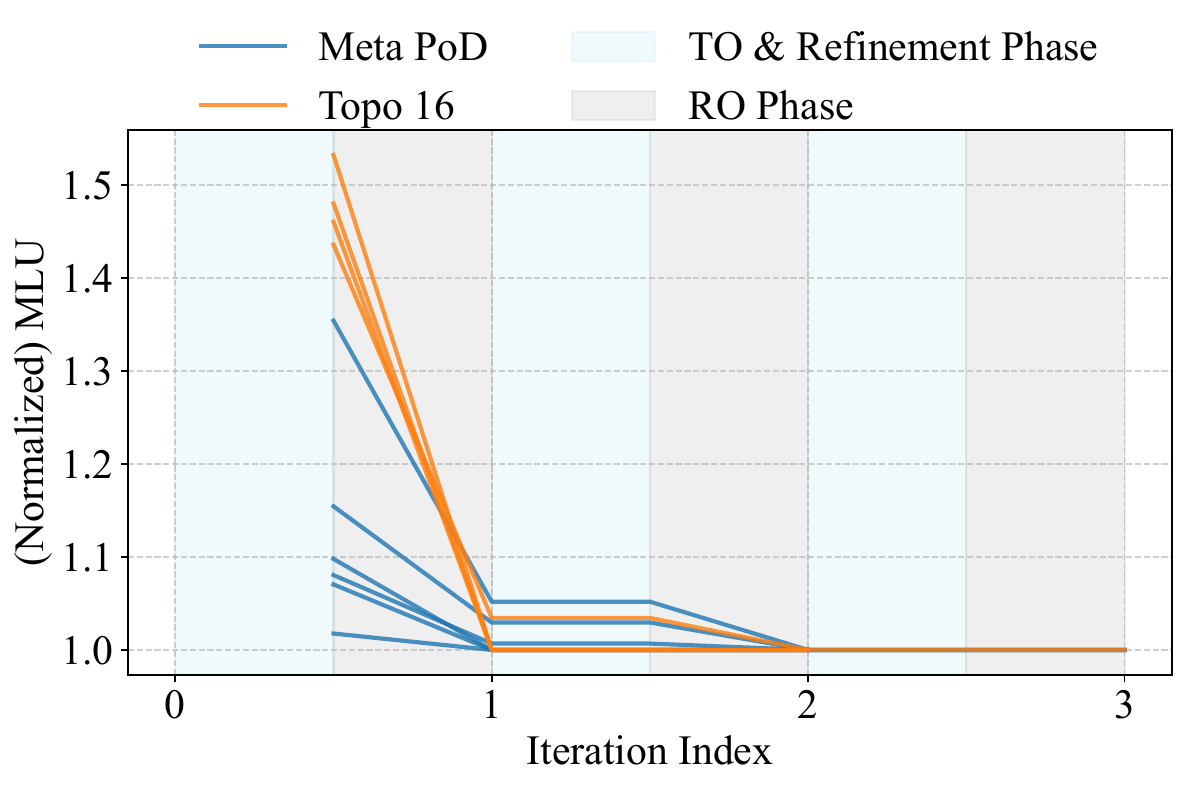}
    \caption{MLU convergence trajectories of ATRO on selected samples from Meta PoD and Topo 16. ATRO converges quickly, often in one or two rounds.}
    \label{fig:convergence}
\end{figure}

\begin{figure}[tbh]
    \centering    
    \includegraphics[width=0.85\linewidth]{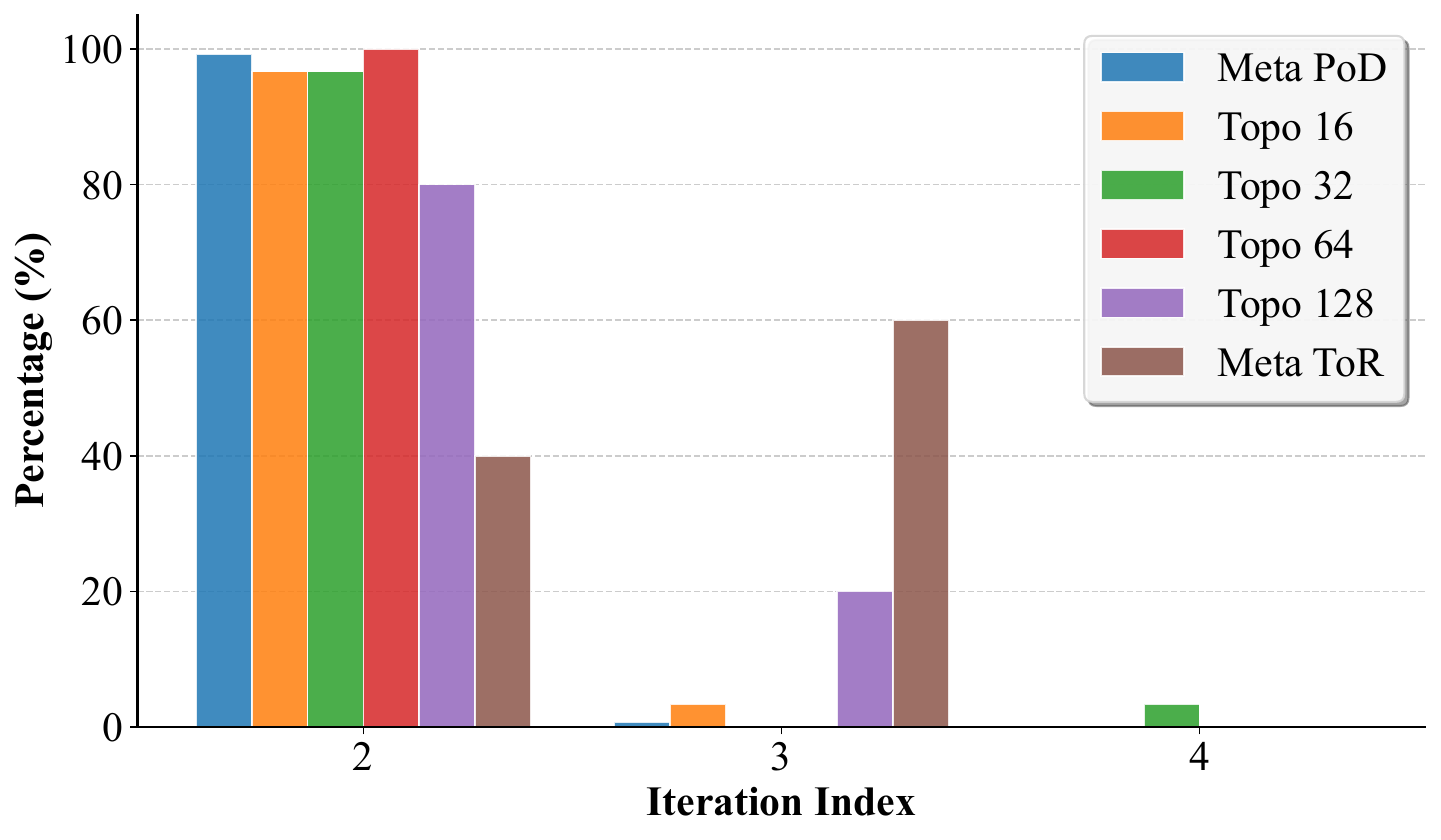}
    \caption{Distribution of convergence rounds across six topologies. Most samples converge within 2–3 iterations, validating the efficiency of the ATRO framework.}
    \label{fig:convergence_round_stats}
\end{figure}

\subsection{Warm Start and Ablation Study}\label{sec:Ablation}

We evaluate the extensibility of ATRO in two aspects: (i) its ability to enhance existing solutions via warm start, and (ii) the effectiveness of its Refinement module.

\noindent\textbf{Warm Start Capability.}  
ATRO accepts any feasible initialization, allowing integration with heuristic or learned solutions. We evaluate two warm-start variants:
\begin{itemize}
    \item \textbf{ATRO-T:} Initialized with COUDER’s topology (bypassing the first TO step).
    \item \textbf{ATRO-R:} Initialized with COUDER’s LP-based routing solution (bypassing the first TO and RO step).
\end{itemize}
As shown in Table~\ref{tab:atro_variants}, both ATRO variants achieve comparable MLU to COUDER across most topologies. On Meta ToR, where the sparse traffic allows many near-optimal solutions, all methods yield similar MLU. Nevertheless, ATRO maintains strong runtime advantages even in this setting.

\noindent\textbf{Effect of the Refinement Component.}  
The \textit{Refinement Component} reallocates residual ports after TO, expanding the solution space and enhancing routing flexibility in subsequent iterations. When this step is disabled—as in the ATRO-O variant—MLU performance degrades in several topologies, particularly those with some AI traffic patterns (e.g., Topo 32 and Topo 64), where routing flexibility is more critical. In contrast, its impact is less pronounced in topologies like Topo 128, where many configurations already satisfy capacity constraints. Nonetheless, Refinement introduces minimal computational overhead and consistently improves robustness, making it a low-cost yet effective enhancement that we recommend by default.

\begin{table}[tbh]
    \caption{Average normalized MLU of COUDER and ATRO variants (values normalized to ATRO).}
    \centering
    \resizebox{\linewidth}{!}{%
        \begin{tabular}{l S S S S}
            \toprule
            Topology & {COUDER} & {ATRO-T} & {ATRO-R} & {ATRO-O} \\
            \midrule
            Meta PoD   & 1.009 & 0.999 & 1.000 & 1.017 \\
            Topo 16    & 1.025 & 1.022 & 1.003 & 1.049 \\
            Topo 32    & 1.660 & 1.287 & 1.027 & 1.007 \\
            Topo 64    & 1.440 & 1.237 & 1.070 & 1.001 \\
            Topo 128   & 1.389 & 1.368 & 1.145 & 1.000 \\
            Meta ToR   & 1.000 & 1.000 & 1.000 & 1.000 \\
            \bottomrule
        \end{tabular}
    }
    \label{tab:atro_variants}
\end{table}







\section{Related Work}

\noindent
\textbf{Reconfigurable DCN Designs.}
Reconfigurable DCNs dynamically adjust logical topologies via optical circuit switches (OCSs) to match traffic demands. Two main architectures have emerged: \textit{multi-hop} and \textit{one-hop}. Multi-hop systems like Jupiter Evolving~\cite{poutievskiJupiterEvolvingTransforming2022} maintain a relatively fixed topology and rely on routing adaptation, requiring joint topology and routing optimization (TRO). In contrast, one-hop systems like RotorNet~\cite{melletteRotorNetScalableLowcomplexity2017} and Sirius~\cite{ballaniSirius2020} reconfigure direct PoD-to-PoD connections per scheduling interval, focusing solely on fast topology optimization (TO). These architectural differences naturally lead to distinct scheduling strategies.

\noindent
\textbf{Scheduling Algorithms for Reconfigurable DCNs.}
In multi-hop systems, COUDER~\cite{tehEnablingQuasiStaticReconfigurable2023} formulates TRO via LP relaxation and rounding. TROD~\cite{caoTRODEvolvingElectrical2021} avoids solvers by first estimating topology based on link-load quantiles and then applying threshold-based splitting for routing decisions; this decoupled approach can be viewed as a heuristic approximation to ATRO, which may lead to capacity waste. One-hop systems prioritize rapid topology computation. Many leverage Birkhoff–von Neumann (BvN) decomposition~\cite{liuSchedulingTechniquesHybrid2015} to express traffic matrices as disjoint matchings, enabling low-latency scheduling. Earlier systems like Helios~\cite{farringtonHeliosHybridElectrical2010} rely on repeated bipartite matchings, incurring high computational overhead that limits scalability.


\section{Conclusion}

We present ATRO, a modular framework for computing logical topologies in reconfigurable data center networks (DCNs). In the general multi-hop setting, ATRO alternates between topology optimization (TO) and routing optimization (RO) using lightweight, scalable subroutines. The TO step is solver-free and solved optimally via our proposed Accelerated Binary Search Method (ABSM), while the RO step supports both LP solvers and TE accelerators—enabling fully solver-free execution if desired. Extensive experiments show that ATRO matches or outperforms existing baselines in both performance and runtime, achieving low-latency scheduling in one-hop settings and efficient scalability in large-scale, multi-hop scenarios. Its convergence-guaranteed, plug-and-play design with warm-start and hybrid support makes ATRO well-suited for real-time, dynamic DCNs. In future work, we will enhance ATRO’s robustness to traffic variations.

\newpage

\bibliographystyle{IEEEtran}
\bibliography{reference}

\end{document}